\documentclass[english,a4paper,twocolumn,superscriptaddress,11pt,accepted=2018-03-07]{quantumarticle}
\usepackage[T1]{fontenc}
\usepackage[latin9]{inputenc}
\usepackage{mathrsfs}
\usepackage{amsmath}
\usepackage{amssymb}
\usepackage{subfig}
\usepackage{graphicx}
\usepackage{color}
\usepackage{hyperref}
\makeatletter

\providecommand{\tabularnewline}{\\}

\newcommand{\1}{{\rm 1\hspace{-0.9mm}l}}

\newcommand{\R}{\ensuremath{\mathbb{R}}}

\newtheorem{prop}{Proposition}
\newtheorem{theorem}{Theorem}

\newtheorem{corollary}{Corollary}

\makeatother

\usepackage{babel}
\begin{document}

\title{Gauge invariant information 
        concerning quantum channels}

\author{ \L ukasz Rudnicki} 
\email{rudnicki@cft.edu.pl}
\affiliation{Institute for Theoretical Physics, University of Cologne, 
Z{\"u}lpicher Stra{\ss}e 77, D-50937, Cologne, Germany}
\affiliation{Center for Theoretical Physics, Polish Academy of Sciences, 
Al. Lotnik{\'o}w 32/46, 02-668 Warsaw, Poland}

\author{Zbigniew Pucha{\l}a}
\affiliation{Institute of Theoretical and Applied Informatics, Polish Academy
of Sciences, ulica Ba{\l}tycka 5, 44-100 Gliwice, Poland}
\affiliation{Faculty of Physics, Astronomy and Applied Computer Science, 
Jagiellonian University, ul. {\L}ojasiewicza 11,  30-348 Krak{\'o}w, Poland}

\author{Karol {\.Z}yczkowski}
\affiliation{Center for Theoretical Physics, Polish Academy of Sciences, 
Al. Lotnik{\'o}w 32/46, 02-668 Warsaw, Poland}
\affiliation{Faculty of Physics, Astronomy and Applied Computer Science, 
Jagiellonian University, ul. {\L}ojasiewicza 11,  30-348 Krak{\'o}w, Poland}

\begin{abstract}
Motivated by the gate set tomography 
we study quantum channels
from the perspective of information which is invariant with respect to   
the gauge realized through similarity of matrices representing channel superoperators. 
We thus use the complex spectrum of the superoperator 
to provide necessary conditions relevant for
complete positivity of qubit channels and to express various metrics
such as average gate fidelity.
\end{abstract}
\maketitle
\section{Introduction} Gauge symmetries are among the most influential
concepts of modern physics. They provide suitable underlying structures
for all known field theories such as Electrodynamics, General Relativity
or Yang-Mills theory. Most importantly, while working with a theory
equipped with a gauge, one becomes guided to meaningful quantities,
namely those which are gauge invariant (do not depend on the gauge
degrees of freedom).

Let us change a perspective for a moment and focus on the following
problem. An unknown square matrix $\Phi$ is provided to us through
a similar matrix $M$ 
\begin{equation}
\Phi=XMX^{-1}\qquad\textrm{or}\qquad M=X^{-1}\Phi X.
\label{gauge}
\end{equation}
Let $X$ be yet another, more or less arbitrary, invertible matrix
of the same size as $\Phi$ and $M$. If $X$ is further assumed to
belong to a certain group, one can be tempted to call (\ref{gauge})
the gauge transformation. It is important to stress that this conceptual
association is right only if we cannot get
the full information concerning $X$, 
so that both $\Phi$ and $M$ are truly equivalent. 

Such a situation is inherent to\emph{ gate set tomography} (GST) \cite{Robin1,Robin2},
a modern view \cite{Greenbaum} on quantum process tomography \cite{QPT}, in which gauge-related terminology is common in use.
In this scheme, the matrix $\Phi$ is a superoperator which represents
a true quantum channel, while $M$ is its reconstruction based on
experimental data. The matrix $X$ is a symmetry, which leaves invariant all
measured probabilities. 

In fact, in GST one simultaneously deals with a, so called, gate set
$\mathcal{G}$ consisting of several gates (quantum channels) $\Phi_{k}$,
an initial state $\rho$ and a final generalized
measurement with two outcomes, 
$\left\{ E,\1-E\right\} $
\textemdash{} all being unknown. During the estimation procedure one
must reconstruct all the gates in question, as well as the initial
state and the applied measurement. The crucial conceptual difference,
with respect to standard quantum (process) tomography, is that linearly
independent sets of initial states and final measurements are obtained
through sequences of gates from $\mathcal{G}$ applied to the single
initial state and the single final measurement. Note that the above idea
was also considered directly in full process-- \cite{Merkel} and state--tomography \cite{Gross}; in the latter case in a somewhat implicit
way. The discussed novelty, while being a source of several advantages
of GST over older approaches, is also responsible for the ``gauge''
ambiguity, since all the probabilities measured during the associated
experiments are of the form:
\begin{equation}
\left\langle \left\langle E\right|\right.\!\Phi_{k_{1}}\Phi_{k_{2}}\cdots\Phi_{k_{n}}\left.\!\left|\rho\right\rangle \right\rangle ,\label{probabilities}
\end{equation}
where $\left.\!\left|\cdot\right\rangle \right\rangle $ and $\left\langle \left\langle \cdot\right|\right.\!$
denote (co)vector representations of involved matrices, to be multiplied
by the superoperators. Clearly \cite{Robin1,Greenbaum}, the probabilities
(\ref{probabilities}) remain the same if instead of $\Phi_{k}$,
$\left.\!\left|\rho\right\rangle \right\rangle $ and $\left\langle \left\langle E\right|\right.\!$
one provides $M_{k}=X^{-1}\Phi_{k}X$, $X^{-1}\left.\!\left|\rho\right\rangle \right\rangle $
and $\left\langle \left\langle E\right|\right.\!X$. 

It seems, however, that so far the ``gauge degree of freedom'' of
the GST,  even though words like  ``transformation'' or  ``invariance'' are among its semantic collocates \cite{Robin3}, has not been treated as a serious player. While working with
Electrodynamics, for instance, one can either consider a vector potential
in a \emph{fixed} gauge (less preferable option) or resort to gauge
invariant quantities such as the electromagnetic field. In GST, the
first path is usually followed, so that $X$ becomes fixed by means
of an optimization procedure. The figure of merit to be optimized
might in principle be a matter of choice or convenience, however,
it usually is expected to minimize the discrepancy between $\mathcal{G}$
and a predefined target gate set $\mathcal{G}'$ \cite{Robin3}. On
the one hand, taking this approach is not surprising since the GST
gauge, in its formal aspects, is only a shadow of its field-theoretic
counterpart. There is no room for structures such as an underlying
manifold (spacetime), gauge connection, curvature, Killing vectors,
etc. On the other hand, by fixing the gauge in a way which makes $\mathcal{G}$
possibly close to $\mathcal{G}'$, one likely underestimates the so-called
QCVV (quantum characterization, verification, and validation) metrics,
such as average gate fidelity \cite{NielsenChuang} or the diamond
norm \cite{diamond}, relevant for fault-tolerant quantum computing. 

A broader context of QCVV enlightens yet another issue. A genuine
gauge is not only a degree of freedom one cannot control (this is
fulfilled in GST) but also one which, being inherent to a physical
phenomenon, cannot be avoided. The latter property is not clearly
satisfied by the GST gauge. Even though one argues that the level
of quantum control, noise reduction etc. necessary to make quantum
process tomography faithful (can one prepare an orthogonal set of
input states in a fully controlled way?) is by far beyond the current
technology \cite{Robin1}, there is no fundamental obstacle in achieving
the desired level in the future.\textbf{ }Moreover, there are other,
complementary methods such as randomized benchmarking 
\cite{EAZ05, RB1, RB2, RBwithConfidence},
which allow for quality assessment of the implemented gates. Thus,
even though there is no room for a proper, rich GST gauge theory,
the main aspect brought by the gauge  --- 
the necessity to
rely only on gauge invariant quantities --- 
seems rather indispensable. 

The aim of this paper is twofold. First of all, we make first attempts
to discuss about properties of quantum channels by means of gauge-invariant
information, showing natural limitations appearing in that strategy.
In more detail, we provide in Section \ref{CPsec} gauge invariant,
necessary criteria for complete positivity of qubit ($d=2$) quantum
channels and discuss in Section \ref{QCVVsec} gauge invariant expressions
and bounds (in arbitrary dimension $d$) for few important QCVV metrics.
Second of all, we would like to point out that unexpectedly (at least
for the authors) there are still interesting and non-trivial problems
associated with the sole description of single qubit channels.

\section{Preliminaries}

From now on, since the gauge transformation acts in the same way on
all gates from $\mathcal{G}$, we restrict our considerations to a
single quantum channel $\mathcal{E}$ acting on $d\times d$ density
matrices. The action of the channel, defined in terms of 
$L\le d^2$ Kraus operators $K_n$,
\begin{equation}
\mathcal{E}\left(\rho\right)=\sum_{n=1}^L K_{n}\rho K_{n}^{\dagger},
\label{Kraus1}
\end{equation}
leads to a  representation of the superoperator
by a square matrix of size $d^2$,
\begin{equation}
\Phi=\sum_{n=1}^L K_{n}\otimes {\bar K}_{n},
\end{equation}
where $\bar K_n$ denotes the complex conjugate of $K_n$.
Given an operator basis $\left\{ B_{0},\ldots,B_{d^{2}-1}\right\} $,
such that $B_{0}=\1/\sqrt{d}$ and $\textrm{Tr}\left(B_{i}^{\dagger}B_{j}\right)=\delta_{ij}$,
the superoperator further acquires the block form 
\begin{equation}
\Phi=\left(\begin{array}{cc}
1 & 0\\
\boldsymbol{k} & T
\end{array}\right),\label{Block}
\end{equation}
provided that we assume the channel $\mathcal{E}$ to be trace preserving.
Up to $d^{2}-1$ dimensional rotations, such basis is formed by the
generalized Pauli matrices. The real matrix $T$ is supposed to act
on generalized Bloch vectors of a quantum state, while the real vector
$\boldsymbol{k}$ accounts for the translation of 
the center of of the set of quantum states and vanishes for unital maps.

Real matrices with the block structure (\ref{Block})  and non-singular $T$ form a group
$G$, being a subgroup of $GL_{d^{2}}\left(\mathbb{R}\right)$, such
that $A\in G$ if $A\in GL_{d^{2}}\left(\mathbb{R}\right)$ and $A_{0j}=\delta_{0j}$
(trace preservation).  If $T$ is singular, trace preservation defines a sub-semigroup of 
 $d^2$-dimensional matrices. A natural requirement, saying that the reconstructed
matrix $M$ from (\ref{gauge}) respects trace preservation enforces gauge matrices to belong to $G$. One can verify this observation by testing the relevant condition
\begin{equation}
\sum_{k,l}X_{0l}A_{lk}\left(X^{-1}\right)_{km}=\delta_{0m},
 \ \ \ m=0,\dots, d^2-1,
\end{equation}
for any  $A$ in the block form (\ref{Block}). After multiplying the above vector by $X$ one
finds
\begin{equation}
X_{00}\delta_{0m}+\sum_{l}X_{0l}A_{lm}=X_{0m}.
\end{equation}
The unique solution of this equation (assuming that $X$ does not
depend on $A$) is $X_{0m}=X_{00}\delta_{0m}$. 
We can set $X_{00}=1$
since for any number $\chi$ the matrices $X$ and $\chi X$ provide
the same gauge transformation.

As a result of the above analysis, all gauge invariant information
about the superoperator is stored in its eigenvalues, since 
\begin{equation}
\lambda_{m}\left(\Phi\right)=\lambda_{m}\left(M\right),
\label{eigenvalues}
\end{equation}
for $m=0,\ldots,d^{2}-1$. By $\lambda_{m}\left(\cdot\right)$ we
denote eigenvalues of a matrix and stick to the convention that $\lambda_{0}\left(\cdot\right)=1$,
which is true for the block form (\ref{Block}). More formally,  considering non-singular $T$,  we
say that only class functions (those which are constant on conjugacy
classes for $G$, defined by the gauge) provide the way towards gauge
invariant description of quantum gates.

Looking at the discussed problem from the perspective of theoretical
quantum mechanics, we can see that the general aim is reduced to studies
of spectrum of $\Phi$. Quite surprisingly, the question posed in
that purely formal way, namely how much do we know about the spectrum
of the superoperator corresponding to a CP TP map,
 has not been answered in a comprehensive manner.
We know that \cite{OSID2004}:
\begin{enumerate}
\item $\lambda_{0}\left(\Phi\right)=1$, due to trace preservation,
\item other $\lambda_{m}\left(\Phi\right)=1$, $m=1,\ldots,d^{2}-1$ are
in general complex, though in practice they are either real or come
in conjugate pairs (so that $\det\Phi$ and $\textrm{Tr}\Phi$ are
real), 
\item All eigenvalues belong to the unit disk, $\left|\lambda_{m}\left(\Phi\right)\right|\leq1$.
\end{enumerate}
These general statements on the spectrum of stochastic quantum
maps  can be considered as an analogue of the Frobenius--Perron theorem \cite{BCSZ09}, which concerns  spectra of stochastic matrices.

Mathematical literature contains also more 
precise and refined  results concerning 
the spectrum of quantum maps.
A quantum operation $\Phi$ is called {\sl irreducible} if there is no
non-trivial face of the cone of quantum states invariant under 
$\Phi$. In such a case the leading eigenvalue $\lambda_0=1$ 
is non-degenerated \cite{EHK78,Ja12}.
Working with the Kraus representation (\ref{Kraus1})
it is possible to show \cite{Fa96}
that the map is irreducible if and only if the set of Kraus operators 
$\{K_1, \dots K_L\}$ has no non-trivial common invariant subspace.

The structure of the {\sl peripheral spectrum}, which contains all eigenvalues of $\Phi$ with modulus one, was studied by Groh in 
the eighties \cite{Gr81} and also more recently
in context of quantum information \cite{WolfCirac,WPG10, Ra17}.
A map is called {\sl primitive} if the peripheral spectrum
has no eigenvalues different than one.
Operational criteria allowing one to establish whether a given map is  primitive were established in \cite{SPGWC10}.
For a primitive map with non-degenerated leading eigenvalue $\lambda_0=1$,
the subleading eigenvalue $\lambda_1$
has modulus  smaller than one, so the {\sl spectral gap},
$\gamma:=1-|\lambda_1|$ is positive.
The size of the spectral gap determines the quantum dynamics,
as it yields an estimate for the rate with which 
any initial state converges to the single invariant state, 
$\rho_0$ associated with $\lambda_0$.
In the case of a generic quantum
map acting on a $d$--dimensional system
the gap is typically large,
 as the modulus of the subleading eigenvalue 
$|\lambda_1|$ scales  \cite{BCSZ09,BSCSZ10} as $1/d$.


\section{Gauge invariant information about complete positivity of a quantum channel}\label{CPsec} 

We now focus on our first question, namely how much information about
complete positivity of a channel can be learned from the sole collection
of eigenvalues of the superoperator. We shall consider the simplest,
but fairly non-trivial case $d=2$, and  use the singular-value decomposition to represent $T=O_{1}\eta O_{2}^{T}$
with $\eta=\textrm{diag}\left(\eta_{1},\eta_{2},\eta_{3}\right)$,
being a diagonal matrix \cite{Geom}. Since $O_{1},O_{2}\in SO(3)$,
the real parameters $\eta_{l}$ can be negative. The singular values
of $T$ are thus given by $s_{k}=\left|\eta_{k}\right|$. 

The only sensible way of comparison between arbitrary collections of singular values ($s_{i}$) and eigenvalues ($\lambda_{i}$) is by means of
majorization relations (weak and log majorization respectively):\begin{subequations}\label{maj}
\begin{equation}
\sum_{i=1}^{k}\left|\lambda_{i}\right|^{\downarrow}\leq\sum_{i=1}^{k}s_{i}^{\downarrow},\qquad k=1,\ldots,N,\label{maj1}
\end{equation}
\begin{equation}
\prod_{i=1}^{k}\left|\lambda_{i}\right|^{\downarrow}\leq\prod_{i=1}^{k}s_{i}^{\downarrow},\qquad k=1,\ldots,N-1,\label{maj2}
\end{equation}
\begin{equation}
\prod_{i=1}^{N}\left|\lambda_{i}\right|^{\downarrow}=\prod_{i=1}^{N}s_{i}^{\downarrow},\label{maj3}
\end{equation}
\end{subequations}applicable to every $N\times N$ matrix. Traditionally,
by ``$\downarrow$'' we denote the decreasing order, so that $s_{k}^{\downarrow}\geq s_{l}^{\downarrow}$
for all $k\leq l$, and the same ordering is applied to the moduli
of eigenvalues.

\subsection{Unital quantum channels} 
We will first discuss the case of a unital channel, with vanishing translation vector $\boldsymbol{k}=0$,
which is known to be completely positive if and only if the Fujiwara-Algoet conditions (FA) \cite{FuAl}
\begin{equation}
1\pm\eta_{3}\geq\left|\eta_{1}\pm\eta_{2}\right|,\label{FA}
\end{equation}
are satisfied. There is no genuine relation between the numbers $\eta_{l}$
(here and later $l=1,2,3$) and the eigenvalues $\lambda_{l}\left(\Phi\right)\equiv\lambda_{l}\left(T\right)$,
except from the determinant condition
\begin{equation}
\det T=\eta_{1}\eta_{2}\eta_{3}=\lambda_{1}\lambda_{2}\lambda_{3}.
\end{equation}
The determinant of $T$ will thus be of special interest for us later
on. To simplify the notation, we omit inside the current section (whenever there is no confusion)
 the matrix argument of involved eigenvalues. 

\begin{table}
\begin{tabular}{|c|c|c|c|c|}
\hline 
$\eta_{1}$ & $\eta_{2}$ & $\eta_{3}$ & $\textrm{sign}\left(\det T\right)$ & $1\pm\eta_{3}\geq\left|\eta_{1}\pm\eta_{2}\right|$\tabularnewline
\hline 
\hline 
$+s_{1}$ & $+s_{2}$ & $+s_{3}$ & $1$ & $1\pm s_{3}\geq\left|s_{1}\pm s_{2}\right|$\tabularnewline
\hline 
$+s_{1}$ & $+s_{2}$ & $-s_{3}$ & $-1$ & $1\mp s_{3}\geq\left|s_{1}\pm s_{2}\right|$\tabularnewline
\hline 
$+s_{1}$ & $-s_{2}$ & $+s_{3}$ & $-1$ & $1\pm s_{3}\geq\left|s_{1}\mp s_{2}\right|$\tabularnewline
\hline 
$+s_{1}$ & $-s_{2}$ & $-s_{3}$ & $1$ & $1\mp s_{3}\geq\left|s_{1}\mp s_{2}\right|$\tabularnewline
\hline 
$-s_{1}$ & $+s_{2}$ & $+s_{3}$ & $-1$ & $1\pm s_{3}\geq\left|s_{1}\mp s_{2}\right|$\tabularnewline
\hline 
$-s_{1}$ & $+s_{2}$ & $-s_{3}$ & $1$ & $1\mp s_{3}\geq\left|s_{1}\mp s_{2}\right|$\tabularnewline
\hline 
$-s_{1}$ & $-s_{2}$ & $+s_{3}$ & $1$ & $1\pm s_{3}\geq\left|s_{1}\pm s_{2}\right|$\tabularnewline
\hline 
$-s_{1}$ & $-s_{2}$ & $-s_{3}$ & $-1$ & $1\mp s_{3}\geq\left|s_{1}\pm s_{2}\right|$\tabularnewline
\hline 
\end{tabular}

\caption{List of all sign combinations for the triple of eta parameters. Note
that whenever the sign of $\det T$ is positive, the two ``$\pm$''
signs present in the FA condition (\ref{FAsingular1})
are the same, while they are reverted
whenever $\det T$ is negative. \label{Tabela1}}

\end{table}

In order to move towards the goal assumed, we express the FA conditions
in terms of the singular values equal to the absolute values of the
parameters $\eta_{l}$. To this end, one needs to examine (see Table
\ref{Tabela1}) all eight combinations of sign choices for the eta
parameters. Put together, this analysis leads to the equivalent form
of Eq. (\ref{FA}) 
\begin{equation}
\begin{cases}
1\pm s_{3}\geq\left|s_{1}\pm s_{2}\right| & \textrm{for }\det T\geq0\\
1\pm s_{3}\geq\left|s_{1}\mp s_{2}\right| & \textrm{for }\det T\leq0
\end{cases}.
\label{FAsingular1}
\end{equation}
To obtain yet another, more friendly form of these conditions, we
examine different orderings between the singular values. For example,
if $s_{1}\geq s_{2}\geq s_{3}$ then the condition with positive determinant
reduces to $1-\left(s_{1}+s_{2}+s_{3}\right)+2s_{3}\geq0$ and $1-\left(s_{1}+s_{2}+s_{3}\right)+2s_{2}\geq0,$
for ``$+$'' and ``$-$'' case respectively. We can see that the
first variant is more restrictive. In this way (going through all
the orderings) one can replace (\ref{FAsingular1}) by
\begin{subequations}
\begin{equation}
1-s_{1}-s_{2}-s_{3}+2\min\left\{ s_{1};s_{2};s_{3}\right\} \geq0,
\end{equation}
for $\det T\geq0$ and
\begin{equation}
1-s_{1}-s_{2}-s_{3}\geq0,
\end{equation}
for $\det T\leq0$.
\label{FAsingular2}
\end{subequations}
The two cases trivially coincide for $\det T=0$ since in this case
$\min\left\{ s_{1};s_{2};s_{3}\right\} =0$. We can also see that
if $\det T>0$ it is possible to assign $s_{1}=s_{2}=s_{3}=1$ (for
instance by $\eta_{1}=\eta_{2}=\eta_{3}=1$), while if $\det T<0$
there is no such possibility. Due to arithmetic-geometric mean inequality one has
$s_{1}s_{2}s_{3}\leq\left(s_{1}+s_{2}+s_{3}\right)^{3}/27$.  
Therefore, assuming that the determinant is negative, 
from (\ref{FAsingular2}) we find that
$s_{1}s_{2}s_{3}\leq 1/27$, and consequently $0\geq\det T\geq-1/27$.
This result was known before (see Corollary 10 from \cite{WolfCirac})
and it is interesting to add that (Example 4 from \cite{WolfCirac})
for every $d$, there exists a channel with $\det T=-\left(d+1\right)^{1-d^{2}}$
equal to $-1/27$ for $d=2$.

So far, we have managed to rewrite the FA conditions in terms of the
singular values of the matrix $T$, supported by $\det T$ which being
function of the eigenvalues is already in desired form. To go further,
we need to relax (\ref{FAsingular2}) with the help of the majorization relations. From the weak majorization (\ref{maj1}) we find
\begin{equation}
-s_{1}-s_{2}-s_{3}\leq-\left(\left|\lambda_{1}\right|+\left|\lambda_{2}\right|+\left|\lambda_{3}\right|\right),
\end{equation}
while due to log majorization (\ref{maj2},\ref{maj3}) we get
\begin{equation}
s_{N}^{\downarrow}\prod_{i=1}^{N-1}s_{i}^{\downarrow}=\left|\lambda_{N}\right|^{\downarrow}\prod_{i=1}^{N-1}\left|\lambda_{i}\right|^{\downarrow}\leq\left|\lambda_{N}\right|^{\downarrow}\prod_{i=1}^{N-1}s_{i}^{\downarrow},
\end{equation}
so that $s_{N}^{\downarrow}\leq\left|\lambda_{N}\right|^{\downarrow}$,
and consequently
\begin{equation}
\min\left\{ s_{1};s_{2};s_{3}\right\} \leq\min\left\{ \left|\lambda_{1}\right|,\left|\lambda_{2}\right|,\left|\lambda_{3}\right|\right\} .
\end{equation}
Making use of the above inequalities we finally obtain:\begin{theorem}
If $\Phi$ is completely positive and unital, then\begin{subequations}\label{CPinv}
\begin{equation}
1-\left(\left|\lambda_{1}\right|+\left|\lambda_{2}\right|+\left|\lambda_{3}\right|\right)+2\min\left\{ \left|\lambda_{1}\right|,\left|\lambda_{2}\right|,\left|\lambda_{3}\right|\right\} \geq0,
\end{equation}
for $\lambda_{1}\lambda_{2}\lambda_{3}\geq0$ and
\begin{equation}
1-\left(\left|\lambda_{1}\right|+\left|\lambda_{2}\right|+\left|\lambda_{3}\right|\right)\geq0,
\end{equation}
 for $\lambda_{1}\lambda_{2}\lambda_{3} < 0$.
\end{subequations}
\end{theorem} The theorem, as giving necessary criteria of complete
positivity, asserts that the range of eigenvalues of $T$ is contained
in the set characterized by the above inequalities. However if $T$
is a normal matrix, so that its singular values are the same as moduli
of the eigenvalues, the criteria also become sufficient.

In Section~\ref{Zbyszek} we show that if a collection of numbers
$\{\lambda_1,\lambda_2,\lambda_3\}$ satisfies the conditions~\eqref{CPinv}, and
moreover all $\lambda_i$ are real, or $\lambda_1\in \R$ and 
$\lambda_2=\overline{\lambda}_3$, then there exists a unital quantum channel with
eigenvalues $\{1,\lambda_1,\lambda_2,\lambda_3\}$.

\subsection{Non-unital channels} \label{NonU}

After we derived the gauge invariant necessary criteria for complete
positivity of unital quantum channels, an obvious way of continuation
shall concern the general case of an arbitrary translation vector 
$\boldsymbol{k}\neq0$. For $d=2$,
the relevant generalization of the FA conditions was known long ago
\cite{Ruskai}. The requirement of gauge invariance, however, cannot
be straightforwardly imposed on this result, because the gauge brings
problems of conceptual nature. It is clear, that the eigenvalues of
$\Phi$ are independent of $\boldsymbol{k}$. In other words, the
whole, \emph{a priori} information about the vector $\boldsymbol{k}$
is lost due to the gauge symmetry. Therefore, one is not able to derive
conditions for complete positivity different than (\ref{CPinv}).
Instead, one can hope that:
\begin{enumerate}
\item the necessary criteria Eq. (\ref{CPinv}) remain valid for non-unital
channels,
\item one can \emph{a posteriori} gain some knowledge about $\boldsymbol{k}$,
by \emph{imposing} complete positivity.
\end{enumerate}
It turns out, that both expectations are met. 

The first one can be best seen from the alternative shape of criteria
from \cite{Ruskai}, derived recently in \cite{Znidar}. In their
Theorem 4.1 it is explicitly shown that for any translation $\boldsymbol{k}$,
the inequalities (\ref{FA}) must be satisfied, and there is an additional
(single) inequality involving $\boldsymbol{k}$. In other words, since
the conditions (\ref{FA}) hold in general, the same applies to (\ref{CPinv}). 

The extra condition [given in \cite{Znidar} also as Eqs. (37) and (38)] involving $\boldsymbol{k}$ can then serve to restrict the range of
this vector, which in the gauge invariant picture is unknown. According
to Eq. (38) from \cite{Znidar}
\begin{equation}
\left\Vert \boldsymbol{k}\right\Vert ^{2}\leq1-\sum_{i=1}^{3}\eta_{i}^{2}+2\eta_{1}\eta_{2}\eta_{3}\equiv1-\sum_{i=1}^{k}s_{i}^{2}+2\det T,
\end{equation}
with $\left\Vert \cdot\right\Vert $ being the standard Euclidean
norm. Due to the majorization property (\ref{maj1}) we immediately
obtain the gauge invariant bound: \begin{corollary} Whenever $\Phi$
is completely positive, then 
\begin{equation}
\left\Vert \boldsymbol{k}\right\Vert ^{2}\leq1-\left|\lambda_{1}\right|^{2}-\left|\lambda_{2}\right|^{2}-\left|\lambda_{3}\right|^{2}+2\lambda_{1}\lambda_{2}\lambda_{3}.
\end{equation}
\end{corollary} The above bound, which restricts the norm of $\boldsymbol{k}$
for completely positive channels seems relatively strong. For instance,
the special case $\lambda_{1}=\lambda_{2}=\lambda_{3}=1$,  implies $\boldsymbol{k}=0$, so
the channel $\Phi$ is forced to be unital.

The remaining condition Eq. (37) from \cite{Znidar}, in its explicit
form $Z\left(\eta\right)=$
\begin{equation}
\left\Vert \boldsymbol{k}\right\Vert ^{4}-2\left\Vert \boldsymbol{k}\right\Vert ^{2}-2\sum_{i=1}^{3}\!\eta_{i}^{2}\!\left(\!2k_{i}^{2}\!-\!\!\left\Vert \boldsymbol{k}\right\Vert ^{2}\right)+q\left(\eta\right)\geq0,\label{Znidaric2}
\end{equation}
with
\begin{align}
q\left(\eta\right)=&\left(1+\eta_{1}+\eta_{2}+\eta_{3}\right)\left(1+\eta_{1}-\eta_{2}-\eta_{3}\right)\times\nonumber \\
&\left(1-\eta_{1}+\eta_{2}-\eta_{3}\right)\left(1-\eta_{1}-\eta_{2}+\eta_{3}\right),
\end{align}
does not allow for a straightforward characterization in terms of
the eigenvalues of $\Phi$. One can however use the same strategy
as before to pass from the eta parameters to the singular values.
The condition (\ref{Znidaric2}) in terms of $s_{k}$ reads:
\begin{equation}
\begin{cases}
Z\left(s\right)\geq0 & \textrm{for }\det T\geq0\\
Z\left(s\right)-16\det T\geq0 & \textrm{for }\det T\leq0
\end{cases},
\end{equation}
with $Z(\cdot)$ defined in Eq. \ref{Znidaric2}. Given the three eigenvalues of $\Phi$ one can (numerically) optimize
the condition with respect to the singular values, taking the majorization
relations (\ref{maj}) as meaningful constraints. 

\section{Gauge invariant bounds on QCVV metrics}\label{QCVVsec}

In the previous section we were concerned with a formal characterization
of complete positivity (for $d=2$) by means of eigenvalues of the
superoperator. Due to the relation (\ref{eigenvalues}), the obtained
results allow for a direct use of the matrix $M$ which is the GST-reconstructed
(from measurement data) superoperator subject to the gauge (\ref{gauge}). 

One of the main practical goals of GST is verification, whether an
experimentally implemented quantum channel is a faithful representation
of a presumed specification (a target gate). To this end one can evaluate
various, relevant metrics such as the diamond distance or average
gate fidelity. A natural question to be asked is how much about such
metrics can one say using the eigenvalues of the superoperator (so
that one can only use the information from $M$)? Our aim now is to
sketch an answer to that question.

In what follows, we will be concerned with a case of an arbitrary
dimension $d$ and the following three metrics:
\begin{itemize}
\item Average gate fidelity \cite{NielsenChuang},
\begin{equation}
\mathcal{F}_{\textrm{avg}}\left(\mathcal{E}\right)=\int d\psi\left\langle \psi\right|\mathcal{E}\left(\left|\psi\right\rangle \left\langle \psi\right|\right)\left|\psi\right\rangle ,
\end{equation}
also related to the average error rate $r\left(\mathcal{E}\right)=1-\mathcal{F}_{\textrm{avg}}\left(\mathcal{E}\right)$,
\item Unitarity \cite{Unitarity} $u\left(\mathcal{E}\right)=$
\begin{equation}
\frac{d}{d-1}\int d\psi\textrm{Tr}\left[\mathcal{E}'\left(\left|\psi\right\rangle \left\langle \psi\right|\right)^{\dagger}\mathcal{E}'\left(\left|\psi\right\rangle \left\langle \psi\right|\right)\right],
\end{equation}
with $\mathcal{E}'\left(\rho\right)=\mathcal{E}\left(\rho\right)-\left[\textrm{Tr}\mathcal{E}\left(\rho\right)\right]\1/\sqrt{d}$,
\item Diamond norm \cite{diamond}, equivalent to completely bounded trace norm, 
\begin{equation}
\epsilon_{\diamond}\left(\mathcal{E}\right)=\frac{1}{2}\sup_{\sigma}\left\Vert \left(\1_{d}\otimes\mathcal{E}\right)\sigma\right\Vert _{1},
\end{equation}
with $\left\Vert \cdot\right\Vert _{1}$ being the trace norm.
In general the supremum is performed over the entire set of the extended states
of an arbitrary dimension, but in analogy to the property of complete 
positivity it is sufficient to take the size of the ancillary
system to be equal to the size $d$ of the principal system \cite{diamond},
so in our case the
extended state $\sigma$ has dimension $d^2$.

\end{itemize}
The average gate fidelity is a standard measure suitable for error
quantification in quantum gates, since it can be efficiently estimated
experimentally with the help of randomized benchmarking. This big
practical advantage is shared with unitarity, the measure designed
for coherent noise. The diamond distance between two channels 
$\epsilon_{\diamond}\left(\mathcal{E}_{1}-\mathcal{E}_{2}\right)$
is of importance for thresholds of fault tolerance in quantum computing.
However, no efficient protocol to measure this quantity
is known, computation of this distance is rather involved and requires semi-definite
programming~\cite{Watrous2013}. Analytical results were obtained 
in the asymptotic case only~\cite{Nechita2016}, for which  the diamond distance
between two random channels was derived.

We shall start with the following, positively surprising observation.
\begin{prop} The average gate fidelity $\mathcal{F}_{\textrm{avg}}\left(\mathcal{E}\right)$
is gauge invariant. \end{prop}\begin{proof} We resort to the well
known formula for average gate fidelity 
\cite{Nielsen,BOSBJ02,Horodeccy}
\begin{equation}
\mathcal{F}_{\textrm{avg}}\left(\mathcal{E}\right)=\frac{\textrm{Tr}\Phi+d}{d\left(d+1\right)}.
\end{equation}
According to (\ref{eigenvalues}), $\textrm{Tr}\Phi=\textrm{Tr}M$.
\end{proof} Note that, even though expressions for higher order moments
(like the variance) of gate fidelity \cite{variance1,variance2} or
the minimal fidelity \cite{minfid} have been studied, it is not possible
to provide their gauge invariant forms. This observation gives the
average gate fidelity a sort of a unique meaning.

Having the gauge invariance of the first metric established, we can
now move to the second and the third one, which definitely are not
invariant. Note however, that these quantities can be bounded in terms
of the average error rate, $r=1-\mathcal{F}_{\textrm{avg}}$.
A generally valid bound for the unitarity,
\begin{equation}
u\left(\mathcal{E}\right)\geq\left[1-\frac{d}{d-1}r\left(\mathcal{E}\right)\right]^{2},\label{unitaritybound1}
\end{equation}
was obtained in Proposition 8 in \cite{Unitarity}, 
while the following bounds for the diamond norm
\begin{equation}
\frac{d+1}{d}r\left(\mathcal{E}\right)\leq\epsilon_{\diamond}\left(\mathcal{E}-\1_{d}\right)\leq\sqrt{d\left(d+1\right)r\left(\mathcal{E}\right)},
\end{equation}
were established as Proposition 9 in \cite{RBwithConfidence}.
The bound (\ref{unitaritybound1})
is saturated if and only if the matrix $T$ in (\ref{Block}) 
is a diagonal scalar matrix.

Since $r\left(\mathcal{E}\right)$ is gauge invariant the above inequalities
provide ready-to-go estimates for both the diamond distance and the
unitarity. More involved lower and upper bounds for the diamond distance
given as functions of average error rate and unitarity have been recently
derived in \cite{Kung,Wallman}. Note that these upper bounds, after
one sets $u\left(\mathcal{E}\right)=1$, are again function of the
gauge invariant average error rate. Here, we shall only quote the
sharpest among the lower bounds \cite{Wallman}
\begin{equation}
\sqrt{\frac{d^{2}-1}{2d^{2}}}\sqrt{u\left(\mathcal{E}\right)-1+\frac{2d}{d-1}r\left(\mathcal{E}\right)}
\ \leq\ 
\epsilon_{\diamond}\left(\mathcal{E}-\1_{d}\right),
\label{diamondLower}
\end{equation}
Quite obviously
the unitarity, which is not gauge invariant, can be lower-bounded by (\ref{unitaritybound1}).
To make a range of available gauge-invariant options a bit wider,
we shall prove the following \begin{prop}\label{propGross} Unitarity
is lower-bounded
\begin{equation}
u\left(\mathcal{E}\right)
\  \geq\ \frac{\sum_{k}\left|\lambda_{k}\left(\Phi\right)\right|^{2}-d}{d\left(d-1\right)},
\label{secondbound}
\end{equation}
with $\lambda_{k}\left(\Phi\right)$, for $k=0,\ldots,d^{2}-1$, being
the eigenvalues of $\Phi$. The bound is tight whenever the channel
is unitary. \end{prop}\begin{proof} By virtue of Proposition 9 from
\cite{Unitarity} one has
\begin{equation}
u\left(\mathcal{E}\right)=\frac{\textrm{Tr}\left[\Phi{}^{\dagger}\Phi\right]-1-\left\Vert \boldsymbol{k}\right\Vert ^{2}}{d^{2}-1},
\end{equation}
with $\boldsymbol{k}$ being the non-unitality vector defined in (\ref{Block}).
Moreover, Proposition 6 from \cite{Unitarity} establishes the bound
$\left\Vert \boldsymbol{k}\right\Vert ^{2}\leq\left(d-1\right)\left[1-u\left(\mathcal{E}\right)\right]$,
which can be directly applied to the previous equation, leading to
\begin{equation}
u\left(\mathcal{E}\right)
\ \geq \
\frac{\textrm{Tr}\left[\Phi{}^{\dagger}\Phi\right]-1-\left(d-1\right)\left[1-u\left(\mathcal{E}\right)\right]}{d^{2}-1},
\end{equation}
We can solve the last inequality for $u\left(\mathcal{E}\right)$
getting
\begin{equation}
u\left(\mathcal{E}\right)
\ \geq \ 
\frac{\textrm{Tr}\left[\Phi{}^{\dagger}\Phi\right]-d}{d\left(d-1\right)}.
\end{equation}
The term $\textrm{Tr}\left[\Phi{}^{\dagger}\Phi\right]$ is the sum
of squared singular values of $\Phi$. Since the singular values always
majorize (see Eq. \ref{maj1}) the moduli of eigenvalues of a given
matrix (and second power is a Shur-convex function) we obtain the
desired bound (\ref{secondbound}).\end{proof} Both (\ref{unitaritybound1})
and (\ref{secondbound}) can further be used to provide gauge invariant
lower bounds for (\ref{diamondLower}).

\section{Discussion}\label{Zbyszek} 

The motivation of this paper stems from the similarity (gauge) relation
(\ref{gauge}) being of relevance for gate set tomography. If one takes the
limitations brought by the gauge seriously, one is directed towards gauge
invariant ways of inferring conclusions. However, contrarily to the standard
case of field theories in physics, while passing to gauge invariant description
we lose accuracy. Even though, one is allowed to argue that the criteria and
bounds derived in Sec. \ref{CPsec} and Sec. \ref{QCVVsec} respectively could
potentially be improved, we shall stress that some loss of information is
unavoidable. To make it evident, we show the following result valid for qubit
gates.

As explained in the preliminary part, the eigenvalues of $\Phi$ are either real
or come in complex-conjugate pairs. For $d=2$ this implies that either all three
non-trivial eigenvalues are real or (without loss of generality):
$\lambda_{1}\left(\Phi\right)=x\in\mathbb{R}$ and
$\lambda_{2}\left(\Phi\right)=z$, $\lambda_{3}\left(\Phi\right)=\overline{z}$ 
with
$z\in\mathbb{C}$ and $\mathrm{Im}z\neq 0$. In the latter case, the CP conditions (\ref{CPinv}), which as explained in Sec. \ref{NonU} are valid for all qubit channels (not only for unital ones), simplify to the form
\begin{equation}
\begin{cases}
1-\left(x+2\left|z\right|\right)+2\min\left\{ x,\left|z\right|\right\} \geq0 & \textrm{for }x\geq0\\
1-\left(\left|x\right|+2\left|z\right|\right)\geq0 & \textrm{for }x\leq0
\end{cases}.
\end{equation}
It can easily be deduced, that since $\left|x\right|\leq1$ and $\left|z\right|\leq1$
the above, in fact, reduces to a single, basic inequality
\begin{equation}
\left|z\right|\leq\frac{1+x}{2},
\label{simple}
\end{equation}
valid for any one-qubit unital quantum operation 
with two complex eigenvalues of the corresponding superoperator. Possible 
values of the eigenvalue $z$, for a given $x$, are plotted in Fig.~\eqref{fig:possible-z}.
\begin{figure}[ht]
\centering
\subfloat[$x=-0.4$]{\includegraphics[width=0.45\linewidth]{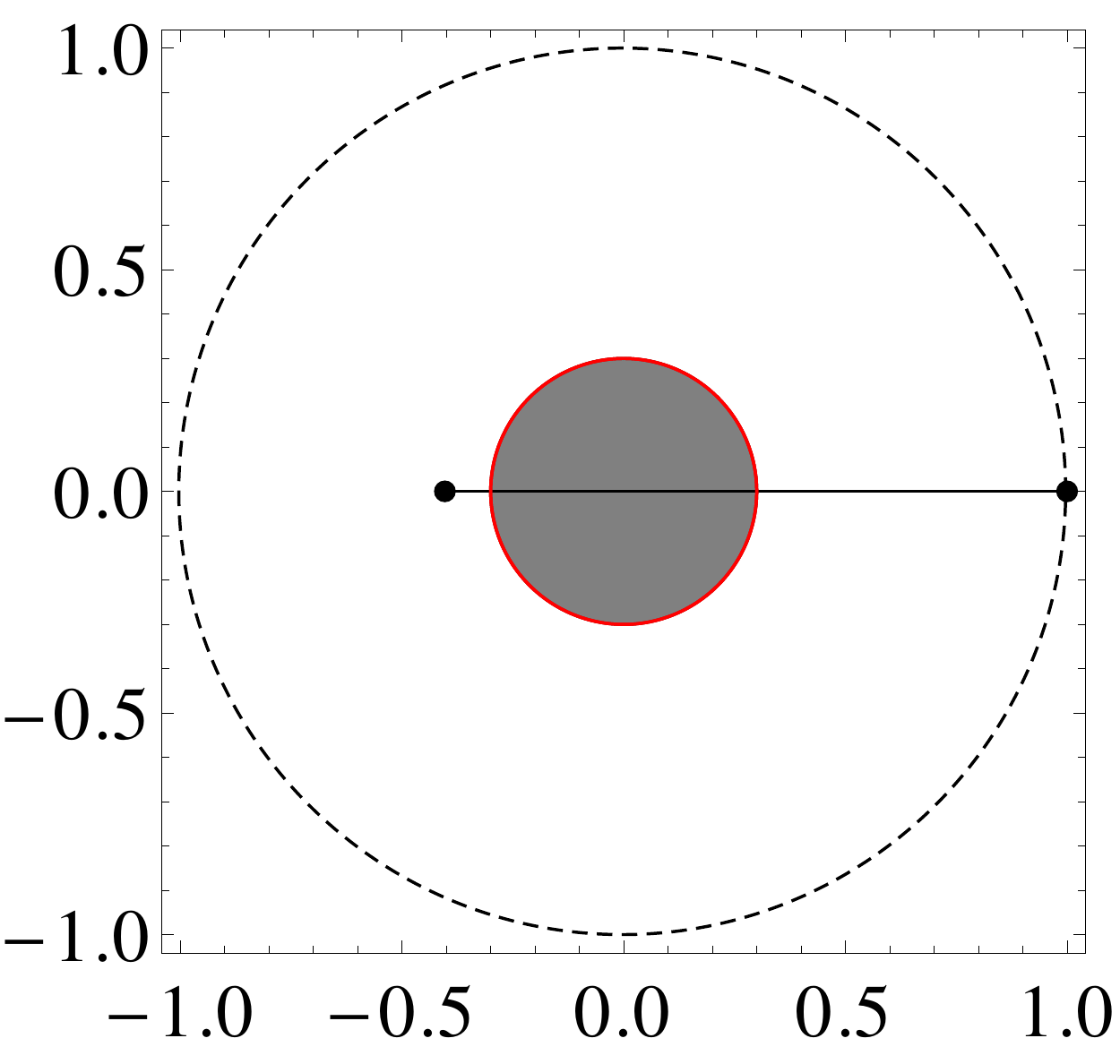}}\quad
\subfloat[$x=0.4$]{\includegraphics[width=0.45\linewidth]{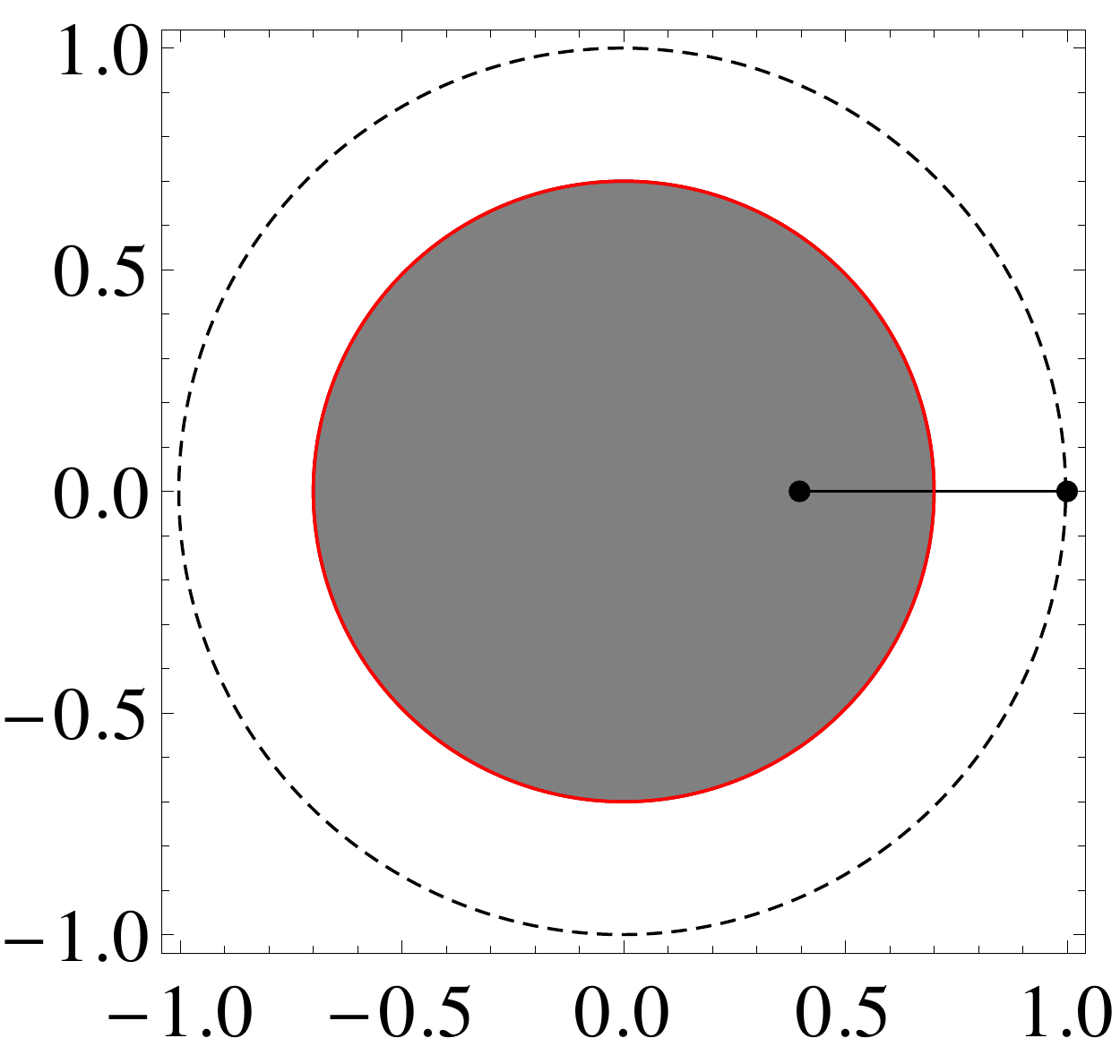}}
\caption{The range of complex eigenvalues for qubit quantum channels. 
In the grey region we plot possible values of $z$ in the case when eigenvalues 
of a qubit channel are in the form $\{1,x,z,\overline{z}\}$, where $x$ is marked at the real axis.
}
\label{fig:possible-z}
\end{figure}

In the case of all three eigenvalues are real the complete positivity 
conditions \eqref{CPinv} are equivalent to the following 
\begin{eqnarray}\label{eqn:real-ev-conditions}
1+\lambda_1+\lambda_2+\lambda_3 &\geq &0,\nonumber\\
1+\lambda_1-\lambda_2-\lambda_3 &\geq &0,\nonumber\\
1-\lambda_1+\lambda_2-\lambda_3 &\geq &0,\nonumber\\
1-\lambda_1-\lambda_2+\lambda_3 &\geq& 0.
\end{eqnarray}
The above states, that if a qubit channel, regardless of whether it is unital or not, has all real eigenvalues, they must lay in 
a tetrahedron defined in~\eqref{eqn:real-ev-conditions}.

Now, let us define two channels, namely: a classical doubly stochastic channel,
described by a superoperator in form of an extended bistochastic matrix,
\begin{equation} S_a =  \left( \begin{smallmatrix} a & 0 & 0 & 1-a \\ 0 & 0 & 0 
& 0 \\ 0 & 0 & 0 & 0 \\ 1-a & 0 & 0 & a  \end{smallmatrix} \right) \ \ \ \text{ 
for } 0\leq a \leq 1, 
\end{equation} 
and a unitary channel with its eigenvectors forming the computational basis 
\begin{equation}
\Psi_\alpha =  \left(
\begin{smallmatrix}
 1 & 0  \\
 0 & e^{i \alpha}  
 \end{smallmatrix} 
 \right) \otimes 
 \left(
 \begin{smallmatrix} 
 1 & 0  \\ 
 0 & e^{- i \alpha}  
 \end{smallmatrix} 
 \right) = 
 \left(
 \begin{smallmatrix} 
 1 & 0 & 0 & 0 \\
 0 & e^{- i \alpha} & 0 & 0 \\
 0 & 0 & e^{i \alpha} & 0 \\
 0 & 0 & 0 & 1
 \end{smallmatrix}
 \right). 
\end{equation}
We are in position to prove the following
\begin{prop}
Let  $\Phi$ be a superoperator of a qubit unital channel characterized by eigenvalues 
$\{1,\lambda_1,\lambda_2,\lambda_3\}$ if:
\begin{itemize}
\item eigenvalues are in the form $\{1,x,z,\overline{z}\}$ with $\mathrm{Im}z\neq 0$ then  $\Phi$ 
is gauge-equivalent to a mixture of a classical map and a unitary 
transformation,
\begin{equation} \label{mixture} 
\Phi_{p,a,\alpha} = p S_a+ (1-p) \Psi_\alpha,
\end{equation}

\item all eigenvalues are real, then $\Phi$ is gauge-equivalent to 
a normal superoperator 
\begin{equation}
\Xi_{\lambda}=
\frac12 \left(
\begin{smallmatrix}
 1+\lambda _1 & 0 & 0 & 1-\lambda _1 \\
 0 & \lambda _2+\lambda _3 & \lambda _3-\lambda _2 & 0 \\
 0 & \lambda _3-\lambda _2 & \lambda _2+\lambda _3 & 0 \\
 1-\lambda _1 & 0 & 0 & 1+\lambda _1
\end{smallmatrix}
\right).
\end{equation}
\end{itemize}

%
\end{prop}
\begin{proof} 
The superoperator (\ref{mixture}) has the eigenvalues: $\{1, 1-2p (1-a), (1-p) 
e^{-i \alpha },(1-p) e^{i \alpha }\}$. They will be equal to those 
of $\Phi$, provided that: $p =  1 - |z|$, $a= \frac{x-2 |z|+1}{2-2  |z|}$ and 
$\alpha = \arg (z)$. This is enough to assure that $\Phi$ and 
$\Phi_{p,a,\alpha} $ are connected by a similarity relation (\ref{gauge}). One 
shall only check that the mixture is proper, namely that $0\leq p, a \leq 1$. 
The first requirement $0\leq p \leq 1$  is trivial since $|z|\leq1$. The second 
one is also simple, i.e. 
\begin{equation}
\begin{split} \frac{x-2 |z| + 1}{2-2|z| } &\geq 0 \ \Longleftrightarrow |z| 
\leq  \frac{1+x}{2}, \\ \frac{x-2 |z| + 1}{2-2 |z| } &\leq 1 \ 
\Longleftrightarrow x \leq 1.
\end{split} 
\end{equation} These requirements are satisfied by virtue of our assumptions, 
in particular, due to the complete positivity condition \eqref{simple}.
In the case of real eigenvalues the proof follows by a direct inspection.
\end{proof}

The above statement is an emanation of the information loss associated
with the gauge-invariant description. Having in our disposal the eigenvalues
of $\Phi$ and nothing more, we are in principle unable to differentiate
between the actual quantum channel and e.g. the mixture (\ref{mixture}),
even if the original channel is clearly distinct from it. 

To shortly conclude, our aim was to check, how much of relevant information
about a quantum channel can be obtained from a collection of 
the eigenvalues of the associated superoperator. 
 As one might not necessarily expect particularly surprising answers to the above question, we shall emphasize practical relevance (within GST) of the problem studied.  It is
clear, that even though the eigenvalues alone contain a substantial
amount of valuable content, the limitation imposed by the GST gauge
invariance treated seriously is quite severe. It thus seems the GST
protocol would gain a lot, provided that one is able to experimentally
infer anything about the gauge degrees of freedom. Quite obviously,
since the GST gauge is not put by the nature on the same footing as
gauge symmetries of fundamental interaction, deeper studies aiming
at providing a routine to fix the gauge in a possibly cleanest way
(maybe not necessarily as the most optimistic one from the QCVV perspective)
are worth further research. 

On the positive side we shall emphasize that the average gate fidelity
or the average error rate are completely specified by the eigenvalues
of the superoperator, accessible by the gate set tomography procedure.
This fact, shall potentially let these quantifiers be
even more useful than they have been by now.

\medskip

\acknowledgments 

We would like to thank David Gross for fruitful discussions and his
support in proving Proposition \ref{propGross}. \L .R. acknowledges
financial support by grant number 2014/13/D/ST2/01886 of the National
Science Center, Poland. Research in Cologne has been supported by
the Excellence Initiative of the German Federal and State Governments
(Grants ZUK 43 \& 81), the ARO under contract W911NF-14-1- 0098 (Quantum 
Characterization, Verification, and Validation), and the DFG projects GRO 
4334/1,2 (SPP1798 CoSIP).
Z.P. acknowledges the support by the Polish National Science Center under
the Project Number 2016/22/E/ST6/00062.
K.{\.Z}. acknowledges the support by the Polish National Science Center under
the Project Number DEC-2015/18/A/ST2/00274 and by the 
John Templeton Foundation under the Project No. 56033.

\end{document}